\documentclass{ectj}
\usepackage{amsfonts,amssymb,graphics,epsfig,verbatim,bm,latexsym,amsmath,url,amsbsy,algorithm,algpseudocode,multirow,threeparttable, booktabs, ragged2e, xcolor, adjustbox, soul}
\newtheorem{theorem}{Theorem}
\newtheorem{assumption}{Assumption}

\newtheorem{lemma}{Lemma}

\newtheorem{remark}{Remark}
\renewcommand{\thesection}{\arabic{section}}
\renewcommand{\theequation}{\arabic{section}.\arabic{equation}}

\DeclareMathOperator{\Var}{Var}
\DeclareMathOperator{\Cov}{Cov}

\newcommand{\E}{\mathbb{E}}
\newcommand{\R}{\mathbb{R}}

\definecolor{inred}{HTML}{CD5C5C}
\definecolor{lgreen}{HTML}{8FBC8F}
\definecolor{bblue}{HTML}{708090}

\newcommand{\txtred}[1]{\textcolor{inred}{#1}}
\newcommand{\txtgreen}[1]{\textcolor{lgreen}{#1}}
\newcommand{\txtblue}[1]{\textcolor{bblue}{#1}}

\received{...}
\accepted{...}
\volume{21}
\title[Double Machine Learning for Time Series]{Double Machine Learning for Time Series}
\author[Ciganovic, D'Amario and Tancioni]{Milos~Ciganovic$^{\dagger}$,
                Federico~D'Amario$^{\ddagger}$ and
                Massimiliano~Tancioni$^{\dagger}$}

\address{$^{\dagger}$Sapienza University of Rome}
\email{milos.ciganovic@uniroma1.it; massimiliano.tancioni@uniroma1.it}

\address{$^{\ddagger}$Bank of England}
\email{Federico.D'Amario@bankofengland.co.uk}
\def\AmSTeX{$\cal A$\kern-.1667em\lower.5ex\hbox{$\cal M$}\kern-.125em
    $\cal S$-\TeX}
\def\BibTeX{{\rm B\kern-.05em{\sc i\kern-.025em b}\kern-.08em
    T\kern-.1667em\lower.7ex\hbox{E}\kern-.125emX}}

\begin{document}
    \begin{abstract}
\sloppy 
We adapt Double Machine Learning to macroeconomic time series by combining regularized nuisance estimation with Reverse Cross-Fitting. This deterministic scheme exploits time reversibility to use time-reversed auxiliary blocks and, unlike neighbor-deletion designs, avoids buffer blocks, thereby improving sample usage. We derive conditions for asymptotic validity and show in simulations that the estimator performs well in realistic finite samples across the designs considered, including cases with misspecification, heteroskedasticity, and state dependence. We also show that, in high dimensions, predictive tuning metrics do not minimize bias in the causal score. We therefore propose a calibration rule targeting a Goldilocks zone of tuning parameters delivering stable partialled-out signals and reduced small-sample bias. We extend the method to residualized Local Projections and apply it to estimate the dynamic effects of a rise in Tier 1 regulatory capital. The results illustrate the usefulness of the approach for macroeconomic time-series inference.
    \end{abstract}
\section{Introduction}
\label{sec:intro}
The recently developed Double/Debiased Machine Learning (DML) estimator (\citealp{Chernozhukov2018}) provides a valid alternative to standard causal inference methods for microeconomic data (\citealp{ahrens2025introduction}). Building on \citet{robinson1988root}'s partially linear regression (PLR) model, it identifies a low-dimensional causal parameter while controlling for potentially high-dimensional confounders. These nuisance functions, though not of direct interest, must be accurately approximated to ensure valid inference, typically through regularization-based estimation in high-dimensional settings. DML reduces the sensitivity of the target parameter estimate to errors in this nuisance approximation. Robustness is achieved through Neyman orthogonalization, which removes regularization bias, and randomized cross-fitting (CF), which alternately assigns folds as training and test samples, mitigating overfitting and improving efficiency in i.i.d. settings.

This paper examines the applicability of DML to macroeconomic time series, which are typically short, strongly dependent, and highly endogenous. Such features can create high-dimensional settings where the parameter space may exceed the sample size. Standard DML, designed for independent observations, cannot yield valid inference in this context, as randomized CF would violate the series's sequential structure.

By modifying specific features of DML, we establish its validity for time-dependent data. The adapted estimator has two components. First, a novel Reverse Cross-Fitting (RCF) procedure uses time
reversibility of stationary Gaussian processes to license the use of
time-reversed auxiliary blocks, and avoids the buffer regions of
neighbor-deletion designs to improve sample usage in realistic time series settings. Second, we introduce a fold-specific, stability-based calibration for nuisance estimators that replaces standard predictive tuning. The criterion targets a “Goldilocks zone” of tuning parameters with locally stable fold-specific RMSE and, consistent with Neyman orthogonality, lower small-sample bias in high-dimensional settings (\citealp{mcgrath2025optimal}).

To parallel the plain DML framework, we formulate the time-series inference problem using the PLR model of \citet{robinson1988root}, as in \citet{Chernozhukov2018}. The goal is root-$T$-consistent estimation and valid inference for a low-dimensional causal parameter $\theta_0$ in the presence of high-dimensional controls and nuisance functions $\eta_0=(m_0,g_0)$:
\begin{align}
    y_t &= \theta_0d_t + g_0(X_t) + \epsilon_t, \quad E[\epsilon_t|X_t,d_t] = 0, \label{y_t_plr}\\
    d_t &= m_0(X_t) + \xi_t, \quad E[\xi_t|X_t] = 0.\label{d_t_plr}
\end{align}
where $y_t$ is the outcome, $d_t$ a scalar policy variable, $X_t = (X_1,\ldots, X_p)$ the vector of controls, and $\epsilon_t$ and $\xi_t$ disturbances, for $t = 1,\ldots,T$.

\citet{belloni2014high} propose post-double-selection, which reduces regularization bias by applying Lasso to the reduced-form of \eqref{y_t_plr} and of \eqref{d_t_plr}, and regressing $y_t$ on $d_t$ and the union of selected controls.\footnote{Naive selection in \eqref{y_t_plr} may omit controls with coefficients moderately close to zero, thereby inducing omitted-variable (regularization) bias (\citealp{belloni2014high}).}

\citet{Chernozhukov2018} generalize this approach to flexible nuisance estimation through DML. Relative to post-double-selection, DML adds three ingredients: (i) Neyman orthogonalization, which removes regularization bias by residualizing the reduced-form outcome equation $y_t=g_0^r(X_t)+\chi_t$, with $g_0^r\neq g_0$, together with \eqref{d_t_plr}, yielding
\begin{align}
\chi_t &= \theta_0 \xi_t + \epsilon_t, \label{resid_eq}
\end{align}
so that applying OLS on residuals accommodates flexible estimation of $g_0^r$ and $m_0$; (ii) sample splitting, which separates nuisance estimation from estimation of $\theta_0$, mitigating overfitting; and (iii) cross-fitting, which rotates the roles of training and test sets to recover sample-use efficiency.

To handle temporal dependence, \citet{semenova2023inference} propose Neighbors-Left-Out (NLO) CF, which partitions the series into blocks and omits neighboring observations to obtain quasi-independent folds. In macroeconomic time series, this can require substantial truncation and reduce sample-use efficiency. We introduce RCF to address this limitation.

We proceed in three steps. First, we develop the time-series RCF-DML estimator and derive conditions for asymptotically valid inference with regularized nuisance estimators in high-dimensional, approximately linear settings. A key requirement is conditional stability, a time-series analog of a no-data-leakage condition that prevents dependence between training and main blocks from reintroducing first-order bias. Under this and standard DML conditions, the estimator is asymptotically consistent, with long-run variance consistently estimable by kernel HAC methods.\footnote{The supplementary appendix reports simulation-based sensitivity checks for alternative HAC kernels and bandwidths.} 

Second, we study the estimator in simulations based on high-dimensional adaptations of standard time-series DGPs. The benchmark is an approximately sparse recursive SVAR (\citealp{giannone2021economic}), complemented by a time-series adaptation of the PLR model (\citealp{belloni2014restud, Chernozhukov2018}) and a structural dynamic factor model (SDFM) following \citet{stockwatson2016DFM}. Across designs, including under model misspecification, RCF-DML delivers root-$T$-consistent estimates with coverage approaching nominal levels and, in realistic finite samples, outperforms NLO. The supplementary simulations further show that these conclusions are broadly robust to severe stress tests of the maintained assumptions, including GARCH-type heteroskedasticity, state dependence, and more demanding factor dynamics. Simulations also show that, in high dimensions, standard predictive tuning does not minimize bias in the causal score, and that the gap between prediction-optimal and minimum-bias tuning widens as dimensionality rises relative to sample size.\footnote{This divergence echoes \citet{robinson1988root}'s findings on allowable kernel complexity in the prototypical PLR model. The issue is independent of the CF design and not driven by failures in relaxing Donsker conditions.} Motivated by this result and related findings in the recent DML literature (\citealp{celentano2023challenges, mcclean2024double, mcgrath2025optimal}), we propose a stability-based tuning rule that targets a Goldilocks zone of hyperparameters with low local variation in fold-specific RMSE. Importantly, RMSE itself need not be minimized within the selected region.\footnote{The benchmark implementation evaluates local RMSE stability over the smallest admissible window. The supplementary appendix also reports sensitivity checks for different stability windows.} Full simulation designs and results are summarized in Section \ref{Simulations} and detailed in the supplementary appendix.

Third, we extend the method to residualized Local Projections (LP) and apply it to estimate the dynamic effects of prudential capital shocks, a setting in which the short span of regulatory data makes RCF-DML especially suitable. Under maintained horizon-specific exogeneity and sufficiently rich high-dimensional conditioning, the time-series-adapted DML estimator provides a feasible procedure for estimating scalar dynamic responses and tracing macroeconomic shock transmission. In this setting, DML implements the required regularized conditioning strategy while preserving valid inference, a distinction that becomes especially important in nonlinear or state-dependent environments. In the application, our focus is on scalar dynamic effects after high-dimensional residualization, rather than on estimating a state-dependent impulse response function.

Our analysis contributes to the growing literature on ML-assisted causal inference (\citealp{athey2016recursive, athey2019generalized, Chernozhukov2018}). Within DML, it relates to recent work on cross-sectional and panel settings (\citealp{ahrens2025model, chang2020double, chernozhukov2022automatic, chernozhukov2022locally, dube2020monopsony, semenova2021debiased, semenova2023inference}). For macroeconomic applications, it complements \citet{adamek2024local}, who study high-dimensional time-series LPs using desparsified Lasso, and \citet{agostini2024vaccination}, who apply dynamic DML-type residualization to panel LPs. On nuisance calibration, our results reinforce the view that tuning for DML inference differs from predictive optimality. \citet{mcclean2024double} advocate undersmoothing to improve convergence, while \citet{mcgrath2025optimal} show that, in high-dimensional (inconsistency) regimes, minimum bias is attained by minimizing the asymptotic variance of the target functional.

The remainder of the paper is organized as follows. Section \ref{RCF} introduces RCF and time reversibility, and establishes asymptotic consistency and efficiency of RCF-DML under suitable conditions. Section \ref{sec:goldilocks} presents the Goldilocks zone tuning rule and its role in bias control for high-dimensional time series. Section \ref{Simulations} reports the simulation evidence. Section \ref{Application} presents the empirical application. Section \ref{conclusions} concludes.

\section{Inference with Reverse Cross-fitting}\label{RCF}
Cross-fitting restores efficiency after sample splitting by rotating training and test roles across folds. In i.i.d. settings this is achieved by random partitioning. However, random splitting is infeasible for time series because it breaks the temporal dependence structure.

We propose RCF as a time-series alternative tailored to macroeconomic data. 
RCF combines two features. The construction of auxiliary samples relies on time reversibility of stationary Gaussian processes, which exploits the use of time-reversed blocks for nuisance estimation; the absence of left-out blocks between auxiliary and main samples maximizes sample usage at moderate number of folds. For realistic macroeconomic samples, this design avoids the truncation costs of NLO.\footnote{A stochastic process $X_t$ is time-reversible if, for any window length $D$, the joint distributions of $X_t = \{X_t,X_{t+1},\ldots, X_{t+D}\}$ and of its time-reversed version, $X'_t = \{X_{t+D},X_{t+D-1},\ldots, X_t\}$ are identical for all $t$. All stationary Gaussian processes are time-reversible. Measurable transformations preserve reversibility because finite-dimensional distributions are determined by the symmetric covariance function (\citealp{weiss1975time,hallin1988time}). More generally, \citet{cheng1999miscellanea} gives necessary and sufficient conditions for reversibility of stationary linear processes without requiring higher-order moments, and \citet{tong2005time} extends the analysis to the multivariate case, covering many non-Gaussian stationary linear processes. Time reversibility is testable; see, for example, \citet{rothman1990characterization} and \citet{martinez2018detection}}

\subsection{The mechanics of time-reversed Cross-Fitting}\label{sec:mechanics}
Let $(W_t)_{t=1}^T$ be a covariance-stationary, ergodic time series on $(\Omega,\mathcal F,\mathbb P)$. Partition the set $\{1, \ldots, T\}$ into adjacent blocks $\{B_k\}_{k=1}^{K}$, each of length $T_{\text{block}}=\lfloor T/K \rfloor$, up to at most $O(1)$ boundary adjustments. For any $k$, define
\[
L(k)=\{1,\ldots,k\}, \qquad R(k)=\{k,\ldots,K\}.
\]
The corresponding left and right auxiliary time-index sets are
\[
B_k^L := \bigcup_{i\in\{1,\ldots,K\}\setminus R(k)} B_i
       = \bigcup_{i<k} B_i,
\qquad
B_k^R := \bigcup_{i\in\{1,\ldots,K\}\setminus L(k)} B_i
       = \bigcup_{i>k} B_i.
\]
Let the associated data sets be
\[
D_k=\{W_{\cdot,t}:t\in B_k\}, \qquad
D_{L,k}=\{W_{\cdot,t}:t\in B_k^L\}, \qquad
D_{R,k}=\{W_{\cdot,t}:t\in B_k^R\}.
\]

To illustrate, consider $K=5$ and, for simplicity, suppose
$T=KT_{\text{block}}$. The partition defines five non-overlapping main
blocks $B_1,\ldots,B_5$ and the corresponding auxiliary index sets:
$B_1^L=\varnothing$, $B_1^R=\bigcup_{i=2}^{5} B_i$;
$B_2^L=B_1$, $B_2^R=\bigcup_{i=3}^{5} B_i$;
$B_3^L=\bigcup_{i=1}^{2} B_i$, $B_3^R=\bigcup_{i=4}^{5} B_i$;
$B_4^L=\bigcup_{i=1}^{3} B_i$, $B_4^R=B_5$;
$B_5^L=\bigcup_{i=1}^{4} B_i$, $B_5^R=\varnothing$.
Nuisance functions are therefore estimated in time-reversed mode in the first two folds, in forward mode in the last two, and in both directions in the central fold.

In typical macroeconomic samples, one side of the auxiliary sample may be too short to support stable nuisance estimation. To avoid nuisance estimates dominated by very short segments, RCF drops the smaller side of the auxiliary sample: $D_{\text{aux},k} = D_{L,k}$ if $|D_{L,k}| > |D_{R,k}|$, $D_{\text{aux},k} = D_{R,k}$ if $|D_{R,k}| > |D_{L,k}|$, and $D_{\text{aux},k} = D_{L,k} \cup D_{R,k}$ if $|D_{L,k}| = |D_{R,k}|$. Figure \ref{fig:RCF} illustrates this five-fold Reverse Cross-Fitting scheme.


\begin{figure}[H]
     \centering
     \includegraphics[width = \textwidth]{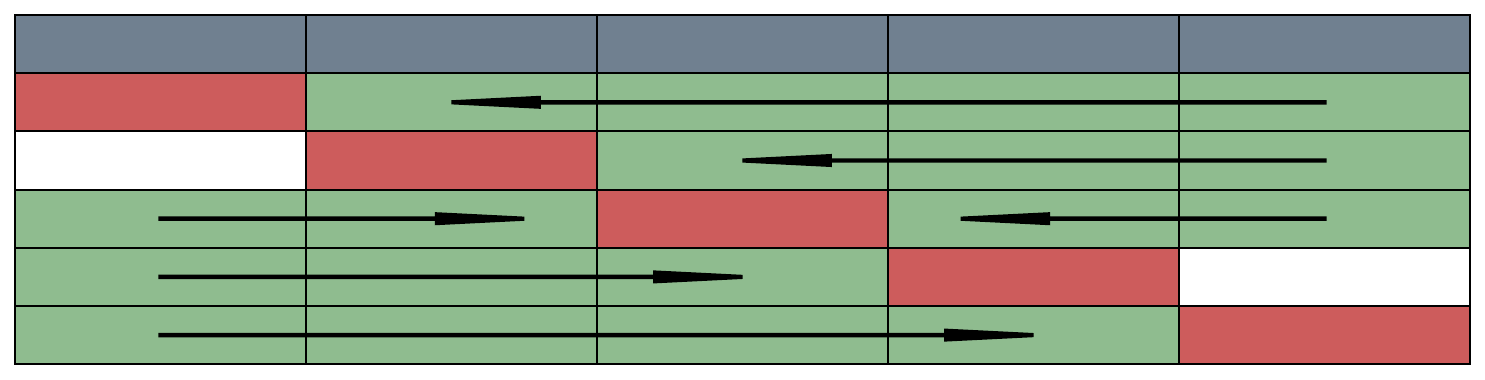}
     \caption{Example of Reverse Cross-Fitting using five folds. The \txtblue{Blue} area represents the whole sample. \txtred{Red} blocks are \txtred{``Main"} observations. \txtgreen{Green} blocks are \txtgreen{``Auxiliary"} observations. White blocks represent left-out observations, and the direction of the arrows indicates the direction of the estimate.}
     \label{fig:RCF}
\end{figure}

To clarify the sample-use advantage of RCF and its efficiency (intended as more effective use of the available sample in the spirit of
\citet{Chernozhukov2018}), compare it with NLO, which omits one neighbor block on each side of the main block (\citealp{semenova2023inference}). Normalizing each fold to one observation, the nuisance-sample usage shares $u$ are
\[
u_{\mathrm{RCF}}(K)=
\begin{cases}
\dfrac{3K-2}{4K}, & K \text{ even},\\[4pt]
\dfrac{3(K^2-1)}{4K^2}, & K \text{ odd},
\end{cases}
\qquad 
u_{\mathrm{NLO}}(K)=\frac{(K-1)(K-2)}{K^2}.
\]

RCF uses more data than NLO for $K=3,\ldots,9$, the two coincide at $K=11$, and NLO uses more data at $K=10$ and for $K\geq 12$. Thus, RCF is not uniformly dominant, but is especially attractive in the moderate-$K$ range relevant for short macroeconomic samples. Its gain comes from two features: time reversibility allows time-reversed auxiliary blocks to target the same population nuisance functions, and the absence of buffer blocks increases usable training observations. The inferential cost of adjacent auxiliary blocks is handled separately by the conditional stability assumption in Section~\ref{sec: assumptions}.

\subsection{The RCF-DML estimator}\label{sec: rcf-dml estimator}
DML targets a finite-dimensional parameter $\theta_0\in\Theta\subset\mathbb R^d$ while controlling for possibly high-dimensional nuisance objects $\eta_0$. Let $\psi:\mathcal W\times\Theta\times\mathcal H\to\mathbb R^d$ be a score function satisfying Neyman orthogonality, the local invariance of $\theta_0$ to errors in the approximation of $\eta_0$:
\begin{equation}
\E[\psi(W_t;\theta_0,\eta_0)] = 0,
\qquad
\partial_\eta \E[\psi(W_t;\theta_0,\eta)]\big|_{\eta=\eta_0} = 0,
\label{eq:orth}
\end{equation}
and let $A := \E[\partial_\theta \psi(W_t;\theta_0,\eta_0)]$ be nonsingular. 
Using the auxiliary sets defined in Section~\ref{sec:mechanics}, RCF trains
the nuisance functions outside the main block \(B_k\). If \(K\) is even, the
auxiliary sample is \(B_k^R\) for \(k\le K/2\) and \(B_k^L\) for \(k>K/2\). If
\(K\) is odd, it is \(B_k^R\) for \(k<(K+1)/2\), \(B_k^L\) for
\(k>(K+1)/2\), and \(B_k^L\cup B_k^R\) for the central fold
\(k=(K+1)/2\). Thus, nuisance estimation and score evaluation are performed on
disjoint observations.

Within $B_k$, use the nuisance estimates $\hat\eta^{(k)}$ to construct residualized outcome and policy series, $(\hat{\chi}_{t,k},\hat{\xi}_{t,k})_{t\in B_k}$, and compute the residual-on-residual OLS slope
\[
\hat\theta_k 
:= 
\arg\min_{\theta\in\Theta} 
\frac{1}{|B_k|}
\sum_{t\in B_k}(\hat{\chi}_{t,k}-\theta\,\hat{\xi}_{t,k})^2.
\]
Equivalently, $\hat\theta_k$ solves the fold-specific empirical moment condition
\[
\frac{1}{|B_k|}\sum_{t\in B_k} \psi(W_t;\hat\theta_k,\hat\eta^{(k)}) =  0
\]

The RCF-DML estimator is the average of the $K$ fold-specific estimates,
\begin{equation}
\hat\theta := \frac{1}{K}\sum_{k=1}^K \hat\theta_k .
\label{eq:thetahat}
\end{equation}

\subsection{Assumptions for valid inference}\label{sec: assumptions}
Let $\psi_t^\star := \psi(W_t;\theta_0,\eta_0)$ denote the oracle score, let
$\Gamma(h):=\Cov(\psi_t^\star,\psi_{t-h}^\star)$,
and define the long-run variance
\begin{equation}
\Sigma := \sum_{h=-\infty}^{\infty}\Gamma(h),
\qquad
\sum_{h=-\infty}^{\infty}\|\Gamma(h)\|<\infty.
\label{eq:Sigma}
\end{equation}
Thus, $\Sigma$ is the time-series analog of $\Var(\psi_t^\star)$ in the i.i.d.\ case.

Valid inference for the RCF-DML estimator relies on the following assumptions.

\begin{assumption}[Finite HAC variance]\label{ass:LRV}
$\E\|\psi_t^\star\|^2<\infty$ and $\Sigma$ in \eqref{eq:Sigma} exists and is finite (equivalently $\sum_{h}\|\Gamma(h)\|<\infty$).
\end{assumption}

Assumption \ref{ass:LRV} rules out explosive or infinite-variance behavior and ensures that the oracle score admits a well-defined HAC variance.

\begin{assumption}[Invariance principle for the oracle score]\label{ass:FCLT}
In the full sample space $D([0,1];\R^d)$,
\[
\frac{1}{\sqrt{T}}\sum_{t=1}^{\lfloor uT\rfloor}\psi_t^\star 
\ \Rightarrow\ 
\mathcal{W}(u),
\]
where $\mathcal{W}$ is a mean-zero Gaussian process with covariance kernel
$\E[\mathcal{W}(u)\mathcal{W}(v)^\top]=(u\wedge v)\,\Sigma$.
\end{assumption}

Assumption \ref{ass:FCLT} is a functional central limit theorem (FCLT) for the oracle score: the partial-sum process of $\psi_t^\star=\psi(W_t;\theta_0,\eta_0)$, normalized by $\sqrt{T}$, converges weakly in the space $D([0,1];\R^d)$ to a multivariate Brownian motion with covariance kernel $(u\wedge v)\Sigma$.\footnote{$u,v\in[0,1]$ denote rescaled time indices, representing fractions of the sample, and $(u\wedge v)\Sigma$ is the covariance kernel of the limiting Brownian motion: the variance grows linearly with time $(u)$, while $\Sigma$ captures the contemporaneous dependence across the $d$ components of the score vector.}
This FCLT yields the asymptotic normality of the RCF-DML estimator under serial dependence (\citealp{Billingsley1968,Hansen1982,Herrndorf1984}). Sufficient conditions include the Maxwell-Woodroofe projective criterion (\citealp{MaxwellWoodroofe2000}) and summable $\alpha$-mixing. Some processes satisfy these weak-dependence conditions together with time reversibility, such as stationary Gaussian linear processes with summable autocovariances. Reversibility, however, is neither necessary nor sufficient for the FCLT and should not be conflated with weak dependence.

\begin{assumption}[Smoothness and Neyman orthogonality]\label{ass:smooth}
\ 

\textbf{(i) Differentiability and orthogonality:} $\psi(W;\theta,\eta)$ is Gateaux-differentiable in $(\theta,\eta)$ near $(\theta_0,\eta_0)$ with integrable derivatives. It admits an $L_2$ expansion in $\eta$ with a quadratic remainder:
\[
\psi(W;\theta_0,\hat\eta)-\psi(W;\theta_0,\eta_0)
=\partial_\eta\psi(W;\theta_0,\eta_0)[\hat\eta-\eta_0]
+ r(W;\hat\eta,\eta_0),
\]
with $\E\|r\|^2\lesssim \|\hat\eta-\eta_0\|_{L_2}^4$. Moreover,
$\E[\partial_\eta\psi(W;\theta_0,\eta_0)[h]]=0$ for all admissible directions $h$ (Neyman orthogonality).

\textbf{(ii) Derivative stationarity and conditional short-memory:} The derivative processes
$\partial_\theta\psi(W_t;\theta_0,\eta_0)$ and $\partial_\eta\psi(W_t;\theta_0,\eta_0)$
are covariance-stationary with finite second moments and summable autocovariances:
\[
\E\|\partial_\theta\psi(W_t;\theta_0,\eta_0)\|^2 < \infty,\quad
\sum_{h}\|\Gamma_{\partial_\theta\psi}(h)\| < \infty,
\]
where $\Gamma_{\partial_\theta\psi}(h) := \Cov(\partial_\theta\psi(W_t;\theta_0,\eta_0), \partial_\theta\psi(W_{t-h};\theta_0,\eta_0))$; and similarly for $\partial_\eta\psi$ (using operator norm). Moreover, for
\[
\mathcal G_t:=\partial_\eta\psi(W_t;\theta_0,\eta_0),
\qquad
\bar{ \mathcal G}_{t,k}:=\E[\mathcal{G}_t\mid\mathcal F_{\mathrm{aux},k}],
\]
the centered $\eta$-derivative array satisfies the conditional short-memory bound
\[
\sum_{t,s\in B_k}
\E\!\left[
\|\mathcal G_t-\bar{\mathcal G}_{t,k}\|_{\mathrm{op}}
\|\mathcal G_s-\bar{\mathcal G}_{s,k}\|_{\mathrm{op}}
\mid \mathcal F_{\mathrm{aux},k}
\right]
=O_p(|B_k|),
\]
uniformly in $k$.

\textbf{(iii) Average Jacobian stability:} For the fold-specific nuisance estimators \(\hat\eta^{(k)}\),
\[
\sup_{1\le k\le K}
\left\|
\frac{1}{|B_k|}\sum_{t\in B_k}
\bigl[
\partial_\theta\psi(W_t;\theta_0,\hat\eta^{(k)})
-
\partial_\theta\psi(W_t;\theta_0,\eta_0)
\bigr]
\right\|
=o_p(1).
\]
For non-affine scores, the same condition holds locally uniformly in
\(\theta\) around \(\theta_0\).

\textbf{(iv) Identification.} The Jacobian $A \;:=\; \E\!\big[\partial_\theta\psi(W_t;\theta_0,\eta_0)\big] \;\in\; \R^{d\times d}$
exists, is finite, and is nonsingular, with $\|A^{-1}\|_{\mathrm{op}}<\infty$.
\end{assumption}

Assumption \ref{ass:smooth} ensures that the score is locally linear in the nuisance parameter and Neyman-orthogonal at $(\theta_0,\eta_0)$, so nuisance estimation errors enter only at second order. Part (i) thus establishes that approximation errors are asymptotically negligible at the $\sqrt{T}$ scale. Part (ii) extends the weak-dependence and continuity requirements to the score derivatives, allowing for valid Jacobian and variance calculations under serial dependence. The conditional short-memory clause is the conditional analog of summable dependence and is verified for the stable linear DGPs considered below. Part (iii) ensures stability of the fold-specific Jacobian, which is needed for
the per-fold linearization. Part (iv) ensures invertibility of the Jacobian.\footnote{The summability of autocovariances parallels the HAC conditions in \citet{Andrews1991} and holds automatically for mixing processes (\citealt{Hansen1982}).}

\begin{assumption}[Dependent cross-fit accuracy and conditional stability]\label{ass:rate}
For each fold $k$,
\[
\|\hat\eta^{(k)}-\eta_0\|_{L_2}=o_p(T^{-1/4}), 
\qquad 
\sup_k\|\hat\eta^{(k)}\|_{L_2}=O_p(1).
\]
Let $\mathcal F_{\mathrm{aux},k}$ denote the $\sigma$-field generated by the auxiliary blocks used to train $\hat\eta^{(k)}$. Then
\[
\Bigg\|
\E\!\Big[
\frac{1}{|B_k|}\sum_{t\in B_k}
\{\psi(W_t;\theta_0,\hat\eta^{(k)})-\psi(W_t;\theta_0,\eta_0)\}
\ \Big|\ \mathcal{F}_{\mathrm{aux},k}
\Big]
\Bigg\| = o_p(T^{-1/2}).
\]
\end{assumption}
Assumption~\ref{ass:rate} combines the standard DML ($L_2$) nuisance-rate condition ($o_p(T^{-1/4})$) and uniform boundedness of fold-specific nuisance estimators with a time-series-specific conditional stability requirement. Because auxiliary and main blocks need not be independent, valid inference requires the block-average conditional plug-in bias, given the training data used to construct $\hat\eta^{(k)}$, to be negligible at the $\sqrt{T}$ scale. This replaces the fold-independence logic used in i.i.d.\ settings. Intuitively, the condition requires that training on non-overlapping but adjacent auxiliary blocks does not induce systematic bias in the scores computed on the main block, the time-series analog of the no-data-leakage principle in statistical learning.

\begin{remark}[On non-independence]\label{remark:non-independence}
RCF does not require the auxiliary and main blocks to be independent. Its baseline implementation uses the available observations outside the main block for nuisance estimation, without neighbor-deletion between training and test samples. Assumption~\ref{ass:rate} makes this permissible: after conditioning on the auxiliary sample used to train the nuisance functions, the score on the main block must differ from the oracle score only through second-order nuisance estimation terms.
\end{remark}

In the PLR case, with scalar policy variable, writing the nuisance errors $\Delta m_t^{(k)}
=
\hat m^{(k)}(X_t)-m_0(X_t)$ and $
\Delta g_t^{r,(k)}
=
\hat g^{r,(k)}(X_t)-g_0^r(X_t)$,
yields the plug-in score expansion (see Supplement S.2)
\[
\psi(W_t;\theta_0,\hat\eta^{(k)})
-
\psi(W_t;\theta_0,\eta_0)
=
-\xi_t\Delta g_t^{r,(k)}
+
(\theta_0\xi_t-\epsilon_t)\Delta m_t^{(k)}
+
\Delta m_t^{(k)}\Delta g_t^{r,(k)}
-
\theta_0\{\Delta m_t^{(k)}\}^2 .
\]
Thus, the condition rules out \emph{first-order leakage}: after conditioning on $X_t$, the auxiliary training sample must not predict the conditional mean of the main-block residual innovations. A sufficient formulation is
\[
\E[\xi_t\mid X_t,\mathcal F_{\mathrm{aux},k}]=0,
\qquad
\E[\epsilon_t\mid X_t,\mathcal F_{\mathrm{aux},k}]=0.
\tag{Conditional Stability}
\]
Under this condition, the first-order terms vanish in conditional expectation, and the remaining quadratic terms are bounded by $
\|\hat m^{(k)}-m_0\|_{L_2}
\|\hat g^{r,(k)}-g_0^r\|_{L_2}
+
|\theta_0|
\|\hat m^{(k)}-m_0\|_{L_2}^2,
$
which is $o_p(T^{-1/2})$ under the usual $o_p(T^{-1/4})$ nuisance rates.

\label{par:cond_stability_scope}
The scope of the condition is broad. It holds for stationary linear processes, such as stable VAR, SVAR, and related approximately linear environments, when conditioning on $X_t$ (estimating $m_0$ and $g^r_0$) absorbs the predictable component relevant for the PLR equations. More generally, it is plausible whenever $X_t$ is a sufficiently rich state vector for the dynamics relevant to $(y_t,d_t)$, including observed state-space systems, observed-regime Markov switching environments and factor structures for which the conditioning set contains sufficiently informative factor proxies. The benchmark SDFM design in the simulations is constructed to satisfy this case approximately.

Conditional stability may fail when auxiliary blocks contain information about persistent components not absorbed by $X_t$. Canonical examples include omitted persistent states, residual serial dependence after conditioning, asymmetric volatility structures in which volatility affects the conditional mean, and contemporaneously endogenous regimes. These cases are discussed formally and tested in Supplement~S2 and S3.

The appropriate remedy for departures from conditional stability depends on the source of the violation. Short-memory leakage can be addressed by neighbor deletion in the spirit of NLO CF, inserting a buffer around each main block with size depending on the persistence of the leakage. If the problem instead reflects omitted persistent states, deleting neighbors is not enough; the conditioning set must be enriched with additional lags, factor estimates, regime indicators, or other state proxies. In practice, buffer sizes can be guided by residual diagnostics that test whether leads or lags of adjacent auxiliary residuals predict the main block residuals. Supplement~S2 provides a more detailed discussion.

\begin{assumption}[Block construction]\label{ass:blocks}
The blocks $\{B_k\}_{k=1}^K$ are adjacent, \\ non-overlapping, and cover $\{1,\dots,T\}$ up to $O(1)$ edge corrections. For each fold $k$, the auxiliary sample excludes all indices in $B_k$.
\end{assumption}

\begin{remark}[Choice of $K$]\label{rem:Kfixed}
In finite samples, the choice of $K$ must balance the size of the main and auxiliary blocks against sample use. In the asymptotic theory, however, $K$ is held fixed, since the large-sample properties of the RCF-DML estimator do not depend on how the sample is partitioned across folds.
\end{remark}

The assumptions above operate at distinct logical levels and may fail separately. Assumptions \ref{ass:LRV}--\ref{ass:smooth} impose weak dependence, smoothness, and orthogonality conditions ensuring asymptotic linearity, asymptotic normality, and HAC-consistent inference; these are generic requirements for time-series DML. Assumption~\ref{ass:rate} is the distinct RCF-specific condition. It controls the conditional plug-in bias created by using adjacent auxiliary and main blocks without buffer deletion. Time reversibility serves a different role: it licenses the construction of time-reversed auxiliary samples targeting the same population nuisance functions, but it neither implies weak dependence nor guarantees conditional stability. These components should therefore be assessed separately.

\subsection{Asymptotic properties}\label{sec: asymp}
We derive the asymptotic distribution of the RCF-DML estimator under serial dependence. The estimator admits an oracle score linear representation, is $\sqrt{T}$-consistent and asymptotically normal, and supports HAC-consistent inference. These results deliver valid inference without fold independence or buffer gaps between training and main blocks. Proofs are collected in Appendix A.

\begin{lemma}[Orthogonality remainder under dependence]\label{lem:remainder}
For each fold $k$, define the block-average discrepancy
\[
R_k \;:=\; \frac{1}{|B_k|}\sum_{t\in B_k}\Big\{\psi(W_t;\theta_0,\hat\eta^{(k)})-\psi(W_t;\theta_0,\eta_0)\Big\}.
\]
Under Assumptions \ref{ass:smooth}--\ref{ass:rate}, $R_k=o_p(T^{-1/2})$ uniformly in $k$. Consequently,
\begin{equation}
\frac{1}{|B_k|}\sum_{t\in B_k}\psi(W_t;\theta_0,\hat\eta^{(k)})
\;=\;
\frac{1}{|B_k|}\sum_{t\in B_k}\psi(W_t;\theta_0,\eta_0)
\;+\;o_p(T^{-1/2}),
\quad k=1,\dots,K,
\label{eq:block-oracle}
\end{equation}
and the average over $k$ also satisfies
\begin{equation}
\frac{1}{K}\sum_{k=1}^K \frac{1}{|B_k|}\sum_{t\in B_k}\psi(W_t;\theta_0,\hat\eta^{(k)})
\;=\;
\frac{1}{K}\sum_{k=1}^K \frac{1}{|B_k|}\sum_{t\in B_k}\psi(W_t;\theta_0,\eta_0)
\;+\;o_p(T^{-1/2}).
\label{eq:avg-block-oracle}
\end{equation}
\end{lemma}

\begin{lemma}[Per-fold OLS linearization]\label{lem:linear}
Let $\hat\theta_k$ satisfy the fold moment equation
$\frac{1}{|B_k|}\sum_{t\in B_k}\psi(W_t;\hat\theta_k,\hat\eta^{(k)})=0$ (residual-on-residual OLS slope on block $B_k$).
Under Assumptions \ref{ass:LRV}, 
\ref{ass:FCLT},
\ref{ass:smooth}, \ref{ass:rate}, \ref{ass:blocks} and Lemma~\ref{lem:remainder},
\[
\sqrt{|B_k|}\,(\hat\theta_k-\theta_0)
\;=\;
-A^{-1}\cdot \frac{1}{\sqrt{|B_k|}}\sum_{t\in B_k}\psi(W_t;\theta_0,\eta_0) \;+\; o_p(1),
\qquad k=1,\dots,K,
\]
\end{lemma}

\begin{theorem}[Fold-average RCF attains oracle asymptotic variance]\label{thm:oracle}
Let $\hat\theta = K^{-1}\sum_{k=1}^K \hat\theta_k$ with $\hat\theta_k$ as above. Under Assumptions \ref{ass:LRV}--\ref{ass:blocks} and the invariance principle (Assumption~\ref{ass:FCLT}),
\[
\sqrt{T}(\hat\theta-\theta_0)\ \Rightarrow\ \mathcal N(0,V),
\qquad V:=A^{-1}\Sigma (A^{-1})^\top,
\]
where $\Sigma$ is the long-run variance in \eqref{eq:Sigma}.
\end{theorem}

Although the main blocks \(B_1,\dots,B_K\) are disjoint, the oracle scores
computed on them inherit the serial dependence of \(\{W_t\}\). The fold-specific
estimators satisfy
\[
\hat\theta_k
\;\approx\;
\theta_0 - A^{-1}\frac{1}{|B_k|}\sum_{t\in B_k}\psi(W_t;\theta_0,\eta_0),
\]
so their average
\(\hat\theta=K^{-1}\sum_{k=1}^K\hat\theta_k\) has asymptotic covariance
determined by the long-run variance of the oracle score process. Valid inference
therefore requires estimating \(\Sigma\) in a way that captures the serial
dependence of the stacked, time-ordered score sequence, including dependence
across adjacent main-block boundaries. We do so with an HAC estimator
constructed from the stacked, time-ordered, cross-fitted score series
\(\{\hat s_t\}_{t=1}^T\), rather than from foldwise variance estimates. A full
description of the variance estimator, together with sensitivity checks for
alternative kernels and bandwidth choices, is reported in the supplementary
appendix.

\section{Stability-based tuning of nuisance parameters}
\label{sec:goldilocks}

To tune nuisance parameters, we target a locally stable region of the predictive error profile rather than the single hyperparameter minimizing predictive error. The idea is to select within a region where predictive performance is good and varies little across nearby penalties, thereby avoiding sharp local minima often associated with overfitting or excessive shrinkage. We refer to this region as the \emph{Goldilocks zone}.

Let \(\Lambda=\{\lambda_1,\ldots,\lambda_M\}\) be an ordered grid of scalar
hyperparameters, with \(\lambda_i<\lambda_{i+1}\). For fold \(k\), let
\(\mathcal A_k\) denote the auxiliary sample and let
\(\mathcal V_k\subset\mathcal A_k\) be an auxiliary validation block adjacent
to \(B_k\). For each \(\lambda_i\), define \(\mathcal R(\lambda_i)\) as the
out-of-sample RMSE obtained by training the nuisance estimate on
\(\mathcal A_k\setminus\mathcal V_k\) and validating it on
\(\mathcal V_k\). Thus, the validation block is held out from the preliminary
fit used to compute \(\mathcal R(\lambda_i)\), and the main block \(B_k\) is
not used for tuning.

For a window size \(S\), define the \(j\)-th local window
\(\mathcal W_j=\{j,j+1,\ldots,j+S-1\}\) for
\(j=1,\ldots,M-S+1\), and let
\begin{equation}
\bar{\mathcal R}_j=\frac{1}{S}\sum_{i\in\mathcal W_j}\mathcal R(\lambda_i),
\qquad
V_j=\frac{1}{S}\sum_{i\in\mathcal W_j}
\big(\mathcal R(\lambda_i)-\bar{\mathcal R}_j\big)^2,
\end{equation}
so that \(\bar{\mathcal R}_j\) measures local predictive performance and
\(V_j\) local instability. To make the two objects comparable, we normalize
them across admissible windows using min--max transformations,
\[
\tilde V_j=\frac{V_j-\min_\ell V_\ell}{\max_\ell V_\ell-\min_\ell V_\ell},
\qquad
\tilde{\bar{\mathcal R}}_j=
\frac{\bar{\mathcal R}_j-\min_\ell \bar{\mathcal R}_\ell}
{\max_\ell \bar{\mathcal R}_\ell-\min_\ell \bar{\mathcal R}_\ell},
\]
and score each window by
\begin{equation}
\mathcal S_j=\tilde V_j+\tilde{\bar{\mathcal R}}_j.
\label{eq:goldilocks}
\end{equation}
The Goldilocks zone index and associated window are
\begin{equation}
j^*=\arg\min_j \mathcal S_j,
\qquad
\mathcal W^*:=\mathcal W_{j^*},
\end{equation}
and the selected hyperparameter is the RMSE minimizer within that window,
\begin{equation}
\lambda^*=\arg\min_{\lambda_i:\, i\in\mathcal W^*}\mathcal R(\lambda_i).
\end{equation}
After \(\lambda^*\) is selected, the final nuisance estimate
\(\hat\eta^{(k)}(\lambda^*)\) used in the score is refit on the auxiliary
sample \(\mathcal A_k\). Hence \(B_k\) is never used for tuning or nuisance
fitting, and all tuning and nuisance-fitting steps remain measurable with
respect to the auxiliary \(\sigma\)-field.

This rule favors hyperparameters that combine good predictive performance with local stability, rather than hyperparameters located at sharp predictive minima. In the benchmark implementation, we set the minimal admissible window size \(S=3\); sensitivity to alternative window sizes is reported in the supplementary appendix.
 
The rule is consistent with the DML framework. First, Neyman orthogonality makes the score locally insensitive to nuisance perturbations, so that small deviations of \(\hat\eta\) from \(\eta_0\) enter only at second order. Selecting a locally stable region of the predictive error profile is consistent with this logic: within a window where \(\mathcal R(\lambda)\) varies little, nearby penalties induce similar predictive behavior. Second, that remainder is controlled by the \(o_p(T^{-1/4})\) rate requirement in Assumption~\ref{ass:rate} (\citealp{Chernozhukov2018}). For penalized linear learners such as Lasso and Elastic Net, this rate is typically attained over a range of regularization parameters. Standard predictive tuning may then select a penalty that remains rate-admissible but is poorly suited to second stage inference, for example by over-shrinking the policy equation and attenuating residual treatment variation in finite samples (\citealp{mcgrath2025optimal}). More generally, nuisance models cannot be too simple, leaving residual confounding, nor too complex, absorbing treatment variation or overfitting noise; this balance is the Goldilocks zone of model complexity (\citealp{fort2019goldilocks, snyder2024goldilocks}), in which \(\hat{m}_0(X_t)\) removes confounding while preserving sufficient policy variation and \(\hat{g}_0^{\,r}(X_t)\) denoises without attenuating policy-induced signals. Selecting within a locally stable region of the predictive error profile targets this zone while remaining inside the rate-admissible interval. Third, the proof architecture conditions on the auxiliary \(\sigma\)-field under the premise that the nuisance estimate used in the score is measurable with respect to the auxiliary sample alone. Because the validation RMSEs \(\mathcal R(\lambda_i)\), the window scores \(\mathcal S_j\), \(j^*\), \(\lambda^*\), and the refitted nuisance estimate \(\hat\eta^{(k)}(\lambda^*)\) are all functions of the auxiliary data only, this premise is preserved, and the guarantees of Theorem~\ref{thm:oracle} continue to apply.

\section{Simulations}\label{Simulations}
In this section, we present simulation evidence on the performance of RCF-DML using the benchmark contemporaneously recursive SVAR DGP
\begin{equation}
    Y_t = \mu + \Phi_1 Y_{t-1} + \varepsilon_t, 
    \qquad 
    \varepsilon_t = P u_t, \qquad 
    u_t \sim \mathcal{N}(0,I_n),
    \label{f_sim_SVAR}
\end{equation}
where $Y_t \in \mathbb{R}^n$ and $P$ is a lower-triangular Cholesky factor encoding contemporaneous impacts. The autoregressive matrix $\Phi_1$ induces approximate sparse local dependence via asymmetric, band-limited decay. We consider $p = 100$ covariance-stationary variables and conduct 10{,}000 Monte Carlo simulations. Details on this DGP and on the alternative designs are provided in Supplement S3.
The outcome is ordered last, so it never enters the policy equation contemporaneously. The benchmark results reported in the main text are based on a mis-specified PLR design, meant to mimic empirical settings in which the contemporaneous structure is not known a priori: the policy variable is placed in the middle of the recursive ordering, while the estimated PLR still conditions on the full contemporaneous control set. Because this set may include variables that are contemporaneously downstream of the policy variable, the PLR coefficient need not coincide with the total recursive SVAR impact response. The benchmark should therefore be interpreted as a misspecification stress test rather than exact recovery of the structural target. For comparison, the supplementary appendix also reports results for the correctly specified design in which the policy variable is placed immediately before the outcome, so that the contemporaneous control set coincides with the truly relevant one.

\subsubsection*{Large sample, low complexity.}
We first evaluate the large-sample performance of the RCF-DML estimator by fixing the sample size at $T = 1{,}000$, which approximates the asymptotic regime. Nuisance functions are estimated via RMSE-tuned Lasso\footnote{Because the SVAR design is low-complexity relative to the sample size, tuning the Lasso penalty via the Goldilocks zone offers no improvement, such that the optimal parameter $\lambda^\star$ is selected using the standard predictive criterion.} , and standard errors use a Newey-West HAC kernel with bandwidth $m = \min(h+1,\, 24)$, where $h$ denotes the LP horizon. In non-dynamic simulations, this bandwidth collapses to $m = 1$.

The results in Table \ref{tab:simulation_results_asymp} show that the RCF-DML estimator attains nominal coverage and remains asymptotically stable, with percentage bias close to 1.5\%. The correct coverage and small bias indicate that the FCLT operates as expected, yielding asymptotic normality. Moreover, performance is stable across the range of fold choices. Taken together, these findings confirm the analytical properties established for the RCF-DML procedure and validate the HAC variance estimation strategy adopted in the asymptotic analysis.

\begin{table}[H]
  \makeatletter
  \long\def\@makecaption#1#2{%
    \vskip\abovecaptionskip
    \sbox\@tempboxa{#1: #2}%
    \ifdim \wd\@tempboxa >\hsize
      #1: #2\par
    \else
      #1: #2\par
    \fi}
  \makeatother
  \caption{\label{tab:simulation_results_asymp}Simulation results for bias (\%) and 95\% coverage (Asymptotic).}
  \centering
    \resizebox{\textwidth}{!}{
    \begin{tabular}{@{}lccccccccc@{}}
      \hline\hline
       & K=4 & K=5 & K=6 & K=7 & K=8 & K=9 & K=10 & K=11 & K=12 \\
      \hline
      Bias & 1.4 & 1.5 & 1.7 & 1.7 & 1.5 & 1.4 & 1.6 & 1.6 & 1.6 \\
      Coverage & 94.0 & 94.1 & 94.1 & 94.3 & 94.2 & 94.3 & 93.8 & 94.2 & 93.4 \\
      \hline\hline
    \end{tabular}
    }
  
  \vspace{0.5em}
  \begin{minipage}{\textwidth}
    \footnotesize
    \renewcommand{\baselineskip}{11pt}
    \textbf{Note:} Bias (\%) and 95\% coverage across values of $K$ for sample size $T=1000$. Results are obtained over 10,000 Monte Carlo replications.
  \end{minipage}
\end{table}

\subsubsection*{Small samples, high complexity.}
We next document the bias-reduction gains achieved by the Goldilocks zone (GZ) tuning rule in finite samples relative to the standard predictive metric (RMSE), and compare the performance of the RCF-DML estimator with the alternative NLO CF strategy. We consider sample sizes $T \in \{50, 100, 200\}$ and different values of $K$. Using the complexity measure $c=p/T$ (\citealt{mcgrath2025optimal}), this design spans a range of dimensionality regimes, from high-dimensional settings ($c=2$ and $c=1$) to moderately complex ones ($c=0.5$).

The results in Table \ref{tab:simulation_results_finite} show clear bias-reduction benefits of the stability-based tuning rule (RCF-GZ), particularly in small samples, where high dimensionality amplifies overfitting risk. On average, RCF-GZ reduces bias by approximately $24\%$ relative to RCF-RMSE. Bias also decreases with larger values of $K$ under both tuning methods, though GZ consistently yields the greater improvement. Coverage is largely insensitive to the tuning criterion and remains close to nominal across values of $K$.

\begin{table}[htb]
  \makeatletter
  \long\def\@makecaption#1#2{%
    \vskip\abovecaptionskip
    \sbox\@tempboxa{#1: #2}%
    \ifdim \wd\@tempboxa >\hsize
      #1: #2\par
    \else
      #1: #2\par
    \fi}
  \makeatother
  \caption{\label{tab:simulation_results_finite}Simulation results for bias (\%) and 95\% coverage (Finite Sample).}
  \centering
    \resizebox{\textwidth}{!}{
    \begin{tabular}{@{}lcccccccccccccccccc@{}}
      \hline\hline
       & \multicolumn{6}{c}{T = 50} & \multicolumn{6}{c}{T = 100} & \multicolumn{6}{c}{T = 200} \\
      \cline{2-7} \cline{8-13} \cline{14-19}
      K & \multicolumn{2}{c}{RCF-GZ} & \multicolumn{2}{c}{RCF-RMSE} & \multicolumn{2}{c}{NLO-RMSE} & \multicolumn{2}{c}{RCF-GZ} & \multicolumn{2}{c}{RCF-RMSE} & \multicolumn{2}{c}{NLO-RMSE} & \multicolumn{2}{c}{RCF-GZ} & \multicolumn{2}{c}{RCF-RMSE} & \multicolumn{2}{c}{NLO-RMSE} \\
      \hline
       & Bias & Cov & Bias & Cov & Bias & Cov & Bias & Cov & Bias & Cov & Bias & Cov & Bias & Cov & Bias & Cov & Bias & Cov \\
      \hline
      4 & 6.9 & 89.1 & 8.5 & 88.6 & 9.0 & 84.9 & 4.1 & 91.1 & 4.8 & 90.9 & 6.1 & 89.4 & 2.0 & 93.4 & 2.2 & 93.3 & 3.7 & 92.0 \\
      5 & 5.9 & 88.6 & 7.3 & 88.1 & 8.3 & 83.4 & 3.8 & 90.6 & 4.5 & 90.2 & 5.8 & 88.1 & 2.0 & 92.5 & 2.4 & 92.5 & 3.4 & 91.7 \\
      6 & 6.1 & 87.8 & 7.8 & 87.3 & 7.9 & 79.9 & 4.0 & 90.3 & 5.2 & 90.0 & 5.4 & 87.7 & 2.2 & 92.8 & 2.6 & 92.5 & 2.7 & 91.5 \\
      7 & 5.6 & 87.4 & 7.3 & 86.6 & 7.9 & 77.2 & 3.4 & 90.4 & 4.5 & 90.2 & 4.7 & 87.1 & 2.0 & 92.2 & 2.5 & 92.2 & 2.4 & 91.0 \\
      8 & 5.7 & 86.9 & 7.9 & 86.6 & 7.8 & 74.4 & 4.5 & 90.1 & 5.8 & 89.3 & 5.3 & 85.5 & 2.7 & 92.4 & 3.2 & 91.9 & 2.8 & 90.6 \\
      9 & 6.3 & 87.1 & 8.4 & 86.3 & 8.7 & 70.2 & 3.5 & 89.2 & 4.7 & 88.9 & 4.4 & 84.5 & 2.2 & 92.3 & 2.8 & 91.9 & 2.5 & 90.8 \\
      10 & 5.5 & 87.2 & 7.6 & 87.1 & 7.2 & 70.4 & 3.6 & 89.3 & 4.9 & 88.4 & 4.3 & 83.8 & 1.9 & 91.9 & 2.6 & 91.3 & 2.1 & 89.9 \\
      11 & 6.5 & 86.4 & 8.6 & 85.9 & 8.2 & 61.1 & 3.3 & 88.7 & 4.8 & 88.2 & 4.2 & 81.6 & 2.1 & 91.7 & 2.7 & 91.3 & 2.2 & 89.5 \\
      12 & 5.0 & 86.9 & 7.5 & 86.5 & 6.7 & 61.9 & 3.5 & 89.1 & 4.9 & 88.3 & 4.0 & 80.9 & 2.2 & 91.4 & 2.8 & 91.0 & 2.3 & 88.6 \\
      \hline\hline
    \end{tabular}
    }
  
  \vspace{0.5em}
  \begin{minipage}{\textwidth}
    \footnotesize
    \renewcommand{\baselineskip}{11pt}
    \textbf{Note:} Bias (\%) and 95\% coverage across values of $K$ and sample sizes ($T$). Results are obtained over 10,000 Monte Carlo replications.
  \end{minipage}
\end{table}

Relative to the NLO CF estimator (NLO-RMSE), RCF-RMSE reduces bias by about 7\% on average, with larger gains in smaller samples. RCF-GZ delivers over a 22\% reduction. Even for $K \ge 10$, when NLO-RMSE remains root-$T$ consistent, the bias under NLO-RMSE is still more than 20\% higher than under RCF-GZ. To separate the role of the tuning rule from that of the RCF design, the Supplement S4.6. reports simulations for NLO-GZ alongside NLO-RMSE. Even under the NLO design, the proposed metric improves accuracy relative to RMSE-based tuning without compromising coverage.\footnote{Although the NLO design aims to produce approximately independent folds, we estimate its standard errors using the same HAC procedure as for RCF-DML to account for possible residual serial dependence.} 

Overall, the results show that RCF can improve on NLO, particularly in small samples. In such cases, the trade-off between sample usage and near-independence becomes binding: for NLO, increasing $K$ leaves too little data in training segments, weakening the benefits of neighbor deletion. As expected, these issues diminish as $T$ grows. Supplement S4.6. also reports two additional robustness checks. First, although the benchmark simulations use a standard Newey-West HAC correction, we assess the sensitivity of coverage and standard errors to alternative kernel and bandwidth choices for the generic HAC estimator, including the procedure recommended by \citet{lazarus2018har}. Second, although the benchmark Goldilocks zone implementation sets the window size to $S=3$, we examine the sensitivity of the results to alternative window widths.

\subsection{Robustness to violations of maintained assumptions}
\label{sec:violations}

\sloppy
We assess the robustness of RCF-DML using three complementary designs, detailed in the supplementary appendix. The first augments the benchmark SVAR with GARCH errors (\citealp{lutk2016SVARGARCH}), breaking time reversibility while preserving the conditional mean structure of the residual innovations. The second considers the state-dependent SVAR of \citet{goncalves2024state}, in which regime-specific dynamics break reversibility under exogenous and lagged-endogenous states, and make conditional stability fail when the state is contemporaneously endogenous. The third uses an SDFM with varying levels of idiosyncratic-error persistence $\rho_e$ (\citealp{stockwatson2016DFM}); as $\rho_e$ rises, residual predictability after conditioning on $X_t$ becomes more persistent, making conditional stability increasingly demanding.

Across these designs, the bias of RCF-DML rises by roughly a factor of two to five relative to the benchmark SVAR, but remains modest in absolute terms: about 6--7\% in large samples (10\% in shorter samples) in the most adverse case (the SDFM with $\rho_e=0.9$) and substantially smaller elsewhere. Coverage remains close to nominal under GARCH and under the exogenous and lagged-endogenous state designs, and deteriorates more visibly only when the state is contemporaneously endogenous or idiosyncratic persistence is high. Consistent with \citet{goncalves2024state} and \citet{kolesarplmoller2025nonlinear}, the estimator recovers the weighted average of regime-specific effects when the state is exogenous or predetermined. The deterioration under contemporaneous state endogeneity is nevertheless milder than in \citet{goncalves2024state}, possibly because high-dimensional controls and sample splitting absorb part of the dependence structure.

A practical implication concerns the comparison with buffer-based CF. Even in designs where conditional stability becomes more problematic, including the contemporaneously endogenous state case, RCF-DML typically continues to deliver lower bias and coverage closer to nominal than NLO in finite samples, so its sample-use advantage is maintained. The main exception is the SDFM with strongly persistent idiosyncratic errors ($\rho_e=0.9$), where NLO with Goldilocks zone tuning performs better. When residual diagnostics (for example, regressions of out-of-sample residuals on neighboring in-sample leads and lags) indicate substantial leakage that adjacent auxiliary blocks cannot absorb, buffer-based splitting becomes preferable; otherwise, RCF remains attractive for short, dependent macroeconomic samples.

\subsection{Double Machine Learning for Local Projections.}
\label{sec:dml_lp}

We use RCF-DML to estimate horizon-specific dynamic responses in a LP framework. In the simulations, the target is the recursive SVAR impulse response. Specifically, we consider the LP analog\footnote{\citet{chernozhukov2021causal} and \citet{agostini2024vaccination} provide early applications of DML to LP settings. For identification and interpretation of LP coefficients, see also \citet{jorda2005}, \citet{teulings2014}, \citet{plagborgmoller2021}, and \citet{montieloleaplagborgmoller2021}.} of the PLR model in \ref{y_t_plr} and \ref{d_t_plr}:

\begin{align}
y_{t+h} &= \theta_h d_t + g_h(X_t) + \epsilon_{t+h},
\label{Y_lp}\\
d_t &= m_0(X_t) + \xi_t,
\qquad E[\xi_t\mid X_t] = 0,
\label{f_lp}
\end{align}
for horizons $h=0,1,\ldots,H$.

At each horizon, RCF-DML re-estimates the fold-specific nuisance function in the reduced-form outcome equation, yielding the residualized relation
\begin{equation}
\hat{\chi}_{t+h} = \theta_h \hat{\xi}_t + \tilde{\epsilon}_{t+h},
\label{res_LP}
\end{equation}
where $\hat{\chi}_{t+h} = y_{t+h} - \hat{g}_h^r(X_t)$, $\hat{\xi}_t = d_t - \hat{m}_0(X_t)$, $g_h^r \neq g_h$ is the reduced-form outcome mapping, and $\tilde{\epsilon}_{t+h} \approx \epsilon_{t+h}$ absorbs nuisance estimation error in finite samples. As in the rest of the paper, nuisance functions are estimated using regularized linear methods.

Causal interpretation of the scalar coefficient $\theta_h$ requires the horizon-specific conditional exogeneity condition
\begin{align}
E[\epsilon_{t+h}\mid X_t,d_t]=0,\qquad h=0,1,\ldots,H,
\label{exo_cond}
\end{align}
which is the dynamic analog of conditional unconfoundedness. Because the outcome is dated at $t+h$, the time-$t$ conditioning set $X_t$ must be rich enough to absorb confounding relevant at horizon $h$. In this respect, high-dimensional controls can strengthen the plausibility of \eqref{exo_cond} by approximating a broader observed information set, and partialling out predictable policy variation linked to macro-financial conditions or anticipation effects. DML should therefore be viewed here as a regularization-based extension of regression-control LPs to high-dimensional settings, delivering feasible conditioning, orthogonality-based bias reduction, and valid inference.

The remaining threats to causal interpretation are not specific to DML, but to LP identification more generally. In particular, the procedure does not by itself eliminate contemporaneous simultaneity between $d_t$ and the outcome, omitted structural shocks not absorbed by $X_t$, or anticipatory responses not proxied by observables. When \eqref{exo_cond} holds, $\theta_h$ is interpretable as the causal impulse response conditional on the observed information set; otherwise, it is a residualized projection coefficient whose interpretation depends on the conditioning set and nuisance specification.\footnote{In nonlinear or state-dependent environments, the scalar residualized LP estimand $\theta_h$ should not be interpreted as a state-dependent impulse response function.} Supplementary simulations adapted from \citet{goncalves2024state} are consistent with the main conclusions of that literature.

\section{Application}\label{Application}
We apply the time-series RCF-DML estimator to estimate the dynamic responses of Italian GDP, private nonfinancial corporations (PNFC) lending, PNFC spreads, and the components of the Tier 1 capital ratio to prudential capital shocks.

\subsection{Background}
\label{back_lit}
Since Basel I, minimum capital requirements have been a central instrument of banking regulation. After the 2008 financial crisis, the reform agenda intensified, culminating in the Basel III package aimed at strengthening bank resilience through tighter capital and liquidity standards (\citealp{basel2010assessment}). In principle, higher capital requirements can reduce risk-taking and default probabilities (\citealp{dewatripont1994theory,holmstrom1997financial}), but they may also tighten credit supply and depress real activity in the short run (\citealp{fbf2010reforming, iif2011cumulative}).

A substantial empirical literature studies the macroeconomic effects of bank capital shocks using alternative identification strategies and capital measures. Despite these differences, the evidence points to a common pattern: tighter capital requirements tend to reduce output and corporate lending in the short run and to increase lending spreads. Table S.2 in Supplement S6 compares the main estimates in the literature.

\subsection{Data and Methodology}
\label{data_and_met}
We assemble the dataset from official sources, following related empirical studies on bank capital shocks (\citealp{noss2016estimating, meeks2017capital, mesonnier2017macroeconomic, kanngiesser2020macroeconomic, conti2023bank, d2026economic}). These contributions guide the selection of variables capturing the main transmission channels from regulatory capital to bank balance sheets and macroeconomic outcomes.

The policy variable is the Tier 1 capital-to-risk-weighted-assets ratio, sourced from the IMF.\footnote{The IMF reports semi-annual data for Italy from 2005:Q2 onward; we interpolate them to obtain a quarterly series. The supplementary appendix reports validation exercises.} The control set combines lagged macro-financial variables with contemporaneous values of fast-moving financial indicators. Stock prices, exchange rates, yields, and yields spreads enter contemporaneously, while slower-moving macroeconomic and banking variables enter only with lags.\footnote{Supplement S5. lists all variables and data sources.} This timing structure is motivated by two considerations. First, fast-moving financial variables proxy information plausibly available within the quarter, in a spirit resembling block-recursive timing assumptions in the VAR literature (\citealp{christiano2005nominal}). Second, entering slower-moving variables only with lags helps avoid conditioning on post-treatment mediators that may lie on the transmission path from prudential capital to subsequent outcomes (\citealp{pearl2009causal}). As a robustness check, Supplement S6.1. also considers an outcome-specific timing scheme that enlarges the contemporaneous conditioning set for bank balance-sheet outcomes.

We include three lags of the controls, the outcome, and the policy variable. All series are seasonally adjusted and tested for stationarity; variables in levels are log-differenced, while ratios and interest rates enter in first differences.

The empirical analysis uses the PLR-LP specification in \eqref{Y_lp} and \eqref{f_lp}, together with the residualized regression in \eqref{res_LP}. As in the rest of the paper, nuisance functions are approximated by Lasso, consistent with an approximately linear, high-dimensional setting. We set the number of folds to $K=6$, which provides a practical balance between main- and auxiliary-block size\footnote{Using six folds yields a main-to-auxiliary sample ratio consistent with standard rule-of-thumb practice.\label{ft:13}} (see Section \ref{RCF}). Fold-specific shrinkage parameters are selected by the Goldilocks zone criterion over a grid of \(100\) values.

Under the maintained LP exogeneity condition \eqref{exo_cond}, the horizon-$h$ coefficients can be interpreted as causal responses. In this application, that assumption is made more plausible by combining contemporaneous fast-moving financial controls with lagged macro-financial variables, although the identifying restriction remains substantive and cannot rule out all omitted shocks or anticipation effects.

For outcomes modeled in differences, we report cumulative responses by estimating, at each horizon $h$, a direct LP for the cumulative residualized outcome (\citealp{jorda2025}):
\begin{equation}
\sum_{h=0}^{H} \hat{\chi}_{t+h} = \Theta_H \, \hat{\xi}_t + \tilde{\epsilon}_{t,H},
\label{cum_lp}
\end{equation}
where, by linearity of OLS, $\Theta_H = \sum_{h=0}^{H}\theta_h$. Direct estimation yields HAC standard errors for $\hat{\Theta}_H$ horizon by horizon, since autocorrelation in $\tilde{\epsilon}_{t,H}$ captures the cross-horizon dependence in $\{\hat{\theta}_h\}_{h=0}^{H}$ without requiring separate estimation of their joint covariance matrix. We retain the benchmark choices used throughout the paper: Goldilocks window size $S=3$ and Newey-West HAC inference with horizon-dependent bandwidth $m=\min(h+1,24)$.

\subsection{Results} \label{Results}
We focus on four transmission channels to real GDP: Tier 1 capital, risk-weighted assets, the real value of new loans to PNFCs, and the PNFC interest-rate spread, measured as the difference between the lending rate on new loans and the three-month Euribor (\citealp{conti2023bank}). Figure \ref{fig:Irf_rcf} reports the corresponding cumulative responses. The PNFC spread temporarily rises, while lending declines, consistent with banks reducing relatively high-risk-weighted exposures in response to tighter capital requirements.

\begin{figure}[htb]%
\centering
\begin{adjustbox}{width= 0.8\textwidth}
   \includegraphics{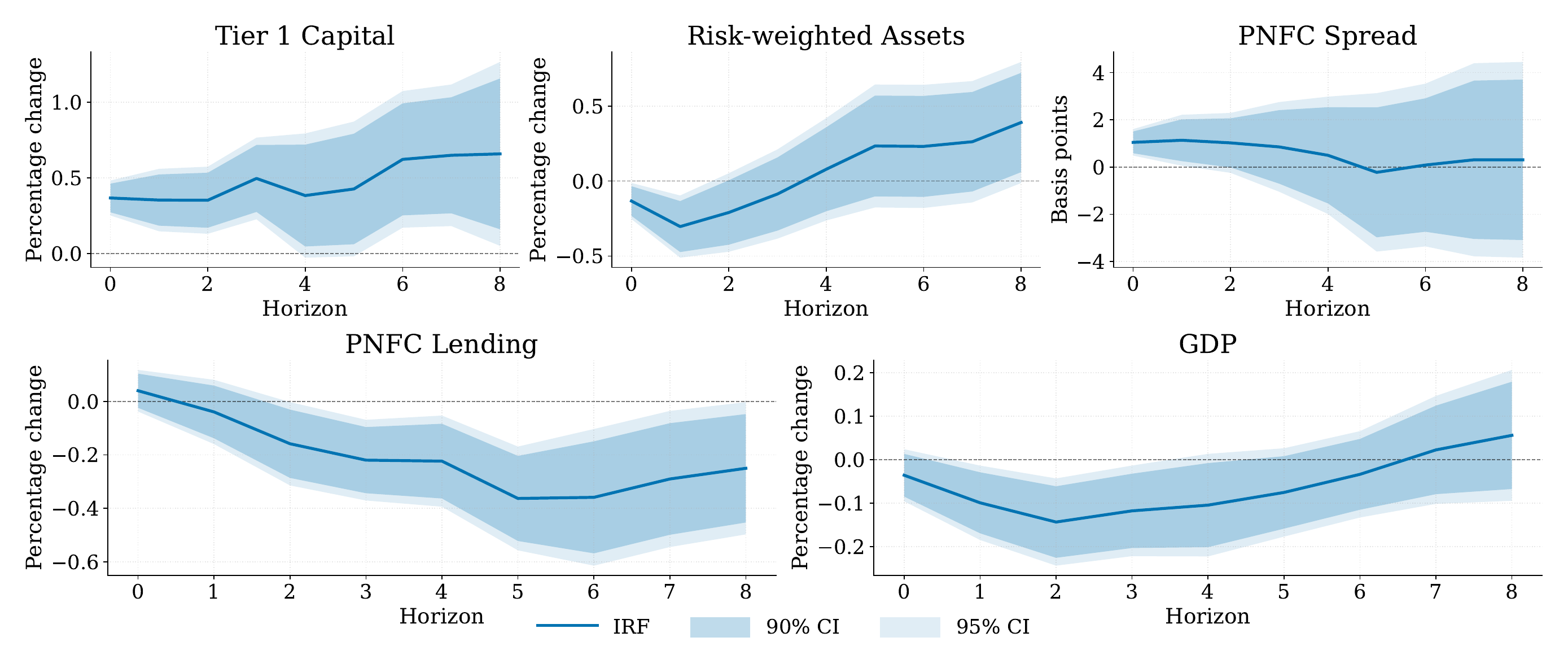}
\end{adjustbox}

\caption{Cumulative impulse response functions to a regulatory capital shock obtained via RCF-DML LPs. The dark (light) blue shaded areas denote 90\% (95\%) cumulative confidence intervals computed via HAC. Responses are normalized to 50 basis points increase in Tier 1 capital ratio. }%

\label{fig:Irf_rcf}%
\end{figure}
The contraction in credit supply is accompanied by adverse short run macroeconomic effects that gradually fade over the medium term. Real GDP declines by almost 0.15 percent after two quarters.

Overall, the responses accord with the transmission patterns emphasized in the empirical literature (Table S.2.). Supplement S6.2. reports the corresponding responses obtained under RMSE-based nuisance calibration, highlighting the value of the proposed tuning rule, especially for GDP. RMSE-based tuning tends to over-denoise the series and attenuate the policy signal, yielding a non-significant GDP response that contrasts with the existing literature.

\section{Conclusions}\label{conclusions}
This paper develops a Double Machine Learning estimator for macroeconomic time series, where short samples, persistence, and endogeneity limit the effectiveness of standard cross-fitting and predictive tuning.

We contribute two methodological innovations. First, Reverse Cross-Fitting exploits time reversibility to construct auxiliary samples from time-reversed blocks and, unlike neighbor-deletion designs, avoids buffer blocks, thereby improving sample-use efficiency while preserving the orthogonality-based conditions required for valid inference. Second, the Goldilocks zone tuning rule is a stability-based criterion that targets hyperparameter regions with stable predictive performance and lower small sample bias. Together, these features improve estimation accuracy in short, dependent samples while preserving the inferential guarantees of Double Machine Learning. We show that the estimator admits an asymptotic linear representation, is $\sqrt{T}$-consistent and asymptotically normal, and that kernel HAC estimators consistently recover the long-run variance.

Monte Carlo experiments support these theoretical results. In approximately sparse time-series designs, Reverse Cross-Fitting yields lower bias and more stable coverage than existing cross-fitting schemes, especially in short samples. The Goldilocks zone rule further improves accuracy and also benefits alternative cross-fitted DML estimators. We also extend the framework to scalar dynamic-response estimation via Local Projections with DML residualization, obtaining limited bias and near-nominal coverage in high-dimensional settings.

Violations of the maintained assumptions, including GARCH-type heteroskedasticity and state dependence, reduce accuracy but do not materially change the main finite-sample conclusions in our simulations. Supplementary simulation evidence further shows that the main findings are not materially affected by alternative HAC kernel and bandwidth choices or by changes in the Goldilocks zone window size.

An empirical application to prudential capital shocks in Italy illustrates the method’s practical relevance. The resulting impulse responses are economically plausible and aligned with the broader empirical literature.

Our approach is best viewed as an alternative, rather than uniformly superior, adaptation of Double Machine Learning to serially dependent data. Theoretical and simulation results clarify the settings in which Reverse Cross-Fitting is most useful.

Overall, the paper provides a feasible and theoretically grounded framework for applying orthogonal score methods with regularized nuisance estimation to macroeconomic time series, where standard DML cross-fitting is unreliable and structural uncertainty is substantial.
\section*{Acknowledgements}
We are grateful to the Co-Editor, Raffaella Giacomini, and anonymous referees for their invaluable comments and suggestions, as well as to Giuseppe Ciccarone, Luca Fanelli, Michele Lenza, Giuseppe Ragusa, and participants of the 2024 Sailing the Macro workshop. The views expressed in this paper are those of the authors and do not necessarily reflect those of the Bank of England.

\setcounter{section}{0}
\renewcommand{\thesection}{A}                 
\setcounter{subsection}{0}
\renewcommand{\thesubsection}{A.\arabic{subsection}}  
\setcounter{equation}{0}
\renewcommand{\theequation}{A.\arabic{equation}}      
\section*{Appendix A: Proofs of Results} \label{Appendix_A}
\subsection*{Proof of LEMMA \ref{lem:remainder}}
\begin{proof}
Throughout, \( |B_k|\asymp T \) by Assumption~\ref{ass:blocks}; all
\(o_p(\cdot)\) rates are uniform in \(k\).

\paragraph{Step 1 (Orthogonal expansion with $L_2$-quadratic remainder).}
By Assumption~\ref{ass:smooth}, for each $t\in B_k$,
\begin{equation}\label{eq:step1-linearization}
\psi(W_t;\theta_0,\hat\eta^{(k)})-\psi(W_t;\theta_0,\eta_0)
= \underbrace{\partial_\eta\psi(W_t;\theta_0,\eta_0)\!\left[\Delta^{(k)}\right]}_{=:L_t^{(k)}}
\;+\; r_t^{(k)},
\qquad \Delta^{(k)}:=\hat\eta^{(k)}-\eta_0,
\end{equation}
with $\E\|r_t^{(k)}\|^2\lesssim \|\Delta^{(k)}\|_{L_2}^4.$

\paragraph{Step 2 (Decomposition into conditional mean and centered part).}
Define the conditional means
\[
\bar L_t^{(k)}:=\E\!\big[L_t^{(k)}\mid \mathcal F_{\mathrm{aux},k}\big],
\qquad
\bar r_t^{(k)}:=\E\!\big[r_t^{(k)}\mid \mathcal F_{\mathrm{aux},k}\big].
\]
Average \eqref{eq:step1-linearization} over $t\in B_k$ and add/subtract the conditional means:
\[
R_k
=
\underbrace{\frac{1}{|B_k|}\sum_{t\in B_k}\big(\bar L_t^{(k)}+\bar r_t^{(k)}\big)}_{\text{(A) conditional mean}}
\;+\;
\underbrace{\frac{1}{|B_k|}\sum_{t\in B_k}\Big[(L_t^{(k)}-\bar L_t^{(k)})+(r_t^{(k)}-\bar r_t^{(k)})\Big]}_{\text{(B) centered fluctuation}}.
\]
We will show: (A) $=o_p(T^{-1/2})$ and (B) $=o_p(T^{-1/2})$.

\paragraph{Step 2.1 (Term A).}
Assumption~\ref{ass:rate} (conditional stability) gives directly
\[
\Bigg\|
\E\!\Big[
\frac{1}{|B_k|}\sum_{t\in B_k}
\{\psi(W_t;\theta_0,\hat\eta^{(k)})-\psi(W_t;\theta_0,\eta_0)\}
\ \Big|\ \mathcal F_{\mathrm{aux},k}
\Big]
\Bigg\| \;=\; o_p(T^{-1/2}),
\]
which equals $\big\|\frac{1}{|B_k|}\sum_{t\in B_k}(\bar L_t^{(k)}+\bar r_t^{(k)})\big\|$ by \eqref{eq:step1-linearization}. Hence (A) $=o_p(T^{-1/2})$.

\paragraph{Step 2.2 (Term B).}
Condition on \(\mathcal F_{\mathrm{aux},k}\) so that
\(\Delta^{(k)}\) is fixed. For the linear part, define
\[
\mathcal G_t:=\partial_\eta\psi(W_t;\theta_0,\eta_0),
\qquad
\bar{\mathcal G}_{t,k}:=
\E[\mathcal G_t\mid \mathcal F_{\mathrm{aux},k}].
\]
Then
\[
L_t^{(k)}-\bar L_t^{(k)}
=
\big(\mathcal G_t-\bar{\mathcal G}_{t,k}\big)[\Delta^{(k)}].
\]
Therefore,
\begin{align}
&\E\!\left[
\left\|
\frac{1}{|B_k|}\sum_{t\in B_k}
\big(L_t^{(k)}-\bar L_t^{(k)}\big)
\right\|^2
\Bigm| \mathcal F_{\mathrm{aux},k}
\right]
\nonumber\\
&\qquad =
\frac{1}{|B_k|^2}
\sum_{t,s\in B_k}
\E\!\left[
\left\langle
\big(\mathcal G_t-\bar{\mathcal G}_{t,k}\big)[\Delta^{(k)}],
\big(\mathcal G_s-\bar{\mathcal G}_{s,k}\big)[\Delta^{(k)}]
\right\rangle
\Bigm| \mathcal F_{\mathrm{aux},k}
\right].
\label{eq:lin-variance}
\end{align}
By Cauchy--Schwarz and the operator norm,
\[
\left|
\left\langle
\big(\mathcal G_t-\bar{\mathcal G}_{t,k}\big)[\Delta^{(k)}],
\big(\mathcal G_s-\bar{\mathcal G}_{s,k}\big)[\Delta^{(k)}]
\right\rangle
\right|
\le
\|\Delta^{(k)}\|_{L_2}^2
\,
\|\mathcal G_t-\bar{\mathcal G}_{t,k}\|_{\mathrm{op}}
\,
\|\mathcal G_s-\bar{\mathcal G}_{s,k}\|_{\mathrm{op}}.
\]

Hence
\begin{align}
&\E\!\left[
\left\|
\frac{1}{|B_k|}\sum_{t\in B_k}
\big(L_t^{(k)}-\bar L_t^{(k)}\big)
\right\|^2
\Bigm| \mathcal F_{\mathrm{aux},k}
\right]\le \nonumber\\
&\qquad\frac{\|\Delta^{(k)}\|_{L_2}^2}{|B_k|^2}
\sum_{t,s\in B_k}
\E\!\left[
\|\mathcal G_t-\bar{\mathcal G}_{t,k}\|_{\mathrm{op}}
\,
\|\mathcal G_s-\bar{\mathcal G}_{s,k}\|_{\mathrm{op}}
\Bigm| \mathcal F_{\mathrm{aux},k}
\right].
\label{eq:HAC-prelim}
\end{align}
By the conditional short-memory part of Assumption~\ref{ass:smooth}(ii),
uniformly in \(k\),
\[
\sum_{t,s\in B_k}
\E\!\left[
\|\mathcal G_t-\bar{\mathcal G}_{t,k}\|_{\mathrm{op}}
\,
\|\mathcal G_s-\bar{\mathcal G}_{s,k}\|_{\mathrm{op}}
\Bigm| \mathcal F_{\mathrm{aux},k}
\right]
=
O_p(|B_k|).
\]
Importantly, this bound does not require pointwise domination of conditional
autocovariances by unconditional autocovariances.
Therefore,
\begin{equation}\label{eq:HAC-G}
\E\!\left[
\left\|
\frac{1}{|B_k|}\sum_{t\in B_k}
\big(L_t^{(k)}-\bar L_t^{(k)}\big)
\right\|^2
\Bigm| \mathcal F_{\mathrm{aux},k}
\right]
\lesssim
\frac{\|\Delta^{(k)}\|_{L_2}^2}{|B_k|}.
\end{equation}
Taking square roots and using
\(\|\Delta^{(k)}\|_{L_2}=o_p(T^{-1/4})\) and \(|B_k|\asymp T\),
\begin{equation}\label{eq:lin-stdev}
\left\|
\frac{1}{|B_k|}\sum_{t\in B_k}
\big(L_t^{(k)}-\bar L_t^{(k)}\big)
\right\|
=
O_p\!\left(
\frac{\|\Delta^{(k)}\|_{L_2}}{\sqrt{|B_k|}}
\right)
=
o_p(T^{-3/4})
=
o_p(T^{-1/2}).
\end{equation}
\paragraph{Step 2.3 (Centered quadratic remainder).}
By the quadratic-remainder part of Assumption~\ref{ass:smooth}(i), interpreted
uniformly for the fold-specific perturbations \(\Delta^{(k)}\),
\[
\E\!\left[
\|r_t^{(k)}\|^2
\Bigm|
\mathcal F_{\mathrm{aux},k}
\right]
\lesssim
\|\Delta^{(k)}\|_{L_2}^4
=
o_p(T^{-1}).
\]
By the triangle inequality and Cauchy--Schwarz,
\[
\left\|
\frac{1}{|B_k|}\sum_{t\in B_k} r_t^{(k)}
\right\|
\le
\frac{1}{|B_k|}\sum_{t\in B_k} \|r_t^{(k)}\|
\le
\left(
\frac{1}{|B_k|}\sum_{t\in B_k} \|r_t^{(k)}\|^2
\right)^{1/2}.
\]
Taking conditional expectations and applying conditional Jensen to the concave
function \(x\mapsto x^{1/2}\),
\[
\E\!\left[
\left\|
\frac{1}{|B_k|}\sum_{t\in B_k} r_t^{(k)}
\right\|
\Bigm|
\mathcal F_{\mathrm{aux},k}
\right]
\le
\left(
\frac{1}{|B_k|}\sum_{t\in B_k}
\E\!\left[
\|r_t^{(k)}\|^2
\Bigm|
\mathcal F_{\mathrm{aux},k}
\right]
\right)^{1/2}
\lesssim
\|\Delta^{(k)}\|_{L_2}^2.
\]
Therefore,
\[
\frac{1}{|B_k|}\sum_{t\in B_k} r_t^{(k)}
=
O_p\!\left(\|\Delta^{(k)}\|_{L_2}^2\right)
=
o_p(T^{-1/2}).
\]
The same bound applies to the conditional mean, since by Jensen's inequality 
applied to the conditional expectation,
\[
\left\|
\frac{1}{|B_k|}\sum_{t\in B_k} \bar r_t^{(k)}
\right\|
=
\left\|
\E\!\left[
\frac{1}{|B_k|}\sum_{t\in B_k} r_t^{(k)}
\Bigm|
\mathcal F_{\mathrm{aux},k}
\right]
\right\|
\le
\E\!\left[
\left\|
\frac{1}{|B_k|}\sum_{t\in B_k} r_t^{(k)}
\right\|
\Bigm|
\mathcal F_{\mathrm{aux},k}
\right]
\lesssim
\|\Delta^{(k)}\|_{L_2}^2.
\]
Hence, by the triangle inequality,
\begin{equation}\label{eq:quad-rem-centered}
\frac{1}{|B_k|}\sum_{t\in B_k}
\big(r_t^{(k)}-\bar r_t^{(k)}\big)
=
O_p\!\left(\|\Delta^{(k)}\|_{L_2}^2\right)
=
o_p(T^{-1/2}).
\end{equation}
\paragraph{Step 2.4 (Term B conclusion).}
Combine \eqref{eq:lin-stdev} and \eqref{eq:quad-rem-centered}:
\[
\frac{1}{|B_k|}\sum_{t\in B_k}\Big[(L_t^{(k)}-\bar L_t^{(k)})+(r_t^{(k)}-\bar r_t^{(k)})\Big]
\;=\; o_p(T^{-1/2}),
\]
uniformly in $k$. Together with Step~2.1, we conclude $R_k=o_p(T^{-1/2})$ uniformly in $k$. 

\paragraph{Step 3 (Conclusion)}
Equation \eqref{eq:block-oracle} follows immediately from \eqref{eq:step1-linearization} and $R_k=o_p(T^{-1/2})$. Averaging \eqref{eq:block-oracle} over $k$ gives \eqref{eq:avg-block-oracle} because $K$ is fixed and the $o_p(T^{-1/2})$ term is uniform in $k$. \hfill $\square$

\end{proof} 

\subsection*{Proof of LEMMA \ref{lem:linear}}
\begin{proof}
Define
\[
\hat A_k(\vartheta)
\;:=\;
\frac{1}{|B_k|}\sum_{t\in B_k}\partial_\theta\psi\!\big(W_t;\vartheta,\hat\eta^{(k)}\big),
\quad
\hat A_k:=\hat A_k(\theta_0),
\quad
A:=\E[\partial_\theta\psi(W_t;\theta_0,\eta_0)].
\]
By Assumption~\ref{ass:smooth}(iv), $A$ is nonsingular.

\paragraph{Step 1 (Mean-value identity in $\theta$, integral form).}
By Gateaux differentiability in $\theta$ (Assumption~\ref{ass:smooth}) and the 
empirical first-order condition 
$|B_k|^{-1}\sum_{t\in B_k}\psi(W_t;\hat\theta_k,\hat\eta^{(k)})=0$, 
the fundamental theorem of calculus applied componentwise yields
\[
0
=
\frac{1}{|B_k|}\sum_{t\in B_k}\psi\!\big(W_t;\theta_0,\hat\eta^{(k)}\big)
\;+\; \tilde A_k\,\big(\hat\theta_k-\theta_0\big),
\qquad
\tilde A_k := \int_0^1 \hat A_k\!\big(\theta_0 + s(\hat\theta_k-\theta_0)\big)\,ds.
\]
Hence
\begin{equation}\label{eq:mv-eq}
\hat\theta_k-\theta_0
=
-\,\tilde A_k^{-1}\cdot
\frac{1}{|B_k|}\sum_{t\in B_k}\psi\!\big(W_t;\theta_0,\hat\eta^{(k)}\big),
\end{equation}
on the event that $\tilde A_k$ is invertible, which has probability approaching 
one by Step~3 below.

\paragraph{Step 2 (Orthogonal plug-in reduction to the oracle score).}
By Lemma~\ref{lem:remainder} (using Assumptions~\ref{ass:smooth} and \ref{ass:rate}),
\[
\frac{1}{|B_k|}\sum_{t\in B_k}\psi\!\big(W_t;\theta_0,\hat\eta^{(k)}\big)
=
\frac{1}{|B_k|}\sum_{t\in B_k}\psi\!\big(W_t;\theta_0,\eta_0\big)
\;+\; o_p\!\big(|B_k|^{-1/2}\big),
\]
uniformly in $k$ (since $K$ is fixed).

\paragraph{Step 3 (Stability of the per-fold Jacobian).}
We show $\tilde A_k\to_p A$. Combined with Assumption~\ref{ass:smooth}(iv) and 
the continuous-mapping theorem, this yields $\tilde A_k^{-1}\to_p A^{-1}$.

Because the PLR score is affine in $\theta$, $\partial_\theta\psi(W_t;\theta,\eta)$ 
does not depend on $\theta$, so $\tilde A_k = \hat A_k(\theta_0)$.\footnote{For a non-affine score, the same argument uses the local-uniform
Jacobian stability condition in Assumption~2.3(iii).} Decompose

\begin{align*}
\hat A_k(\theta_0)-A
&=
\underbrace{
\frac{1}{|B_k|}\sum_{t\in B_k}
\big\{\partial_\theta\psi(W_t;\theta_0,\hat\eta^{(k)})
      -\partial_\theta\psi(W_t;\theta_0,\eta_0)\big\}
}_{(I)} \\
&\quad+
\underbrace{
\frac{1}{|B_k|}\sum_{t\in B_k}
\big\{\partial_\theta\psi(W_t;\theta_0,\eta_0)-A\big\}
}_{(II)} .
\end{align*}
For $(II)$, Assumption~\ref{ass:smooth}(ii) implies $(II)=o_p(1)$ uniformly in
\(k\), since \(K\) is fixed and \(|B_k|\asymp T\). For $(I)$,
Assumption~\ref{ass:smooth}(iii) gives $(I)=o_p(1)$ uniformly in \(k\).
Hence \(\tilde A_k=\hat A_k(\theta_0)=A+o_p(1)\) and
\(\tilde A_k^{-1}\to_p A^{-1}\), uniformly in \(k\).

\paragraph{Step 4 (Consistency and rate).}
Insert Step~2 into \eqref{eq:mv-eq}:
\[
\hat\theta_k-\theta_0
=
-\,\tilde A_k^{-1}\cdot\bigg(
\frac{1}{|B_k|}\sum_{t\in B_k}\psi\!\big(W_t;\theta_0,\eta_0\big)
\;+\; o_p\!\big(|B_k|^{-1/2}\big)\bigg).
\]
Since $\tilde A_k^{-1}=O_p(1)$ by Step~3 and, by Assumption~\ref{ass:FCLT} 
(CLT at $u=1$ on block $B_k$),
\[
\frac{1}{|B_k|}\sum_{t\in B_k}\psi\!\big(W_t;\theta_0,\eta_0\big) \;=\; O_p\!\big(|B_k|^{-1/2}\big),
\]
we obtain
\[
\hat\theta_k-\theta_0 \;=\; O_p\!\big(|B_k|^{-1/2}\big).
\]

\paragraph{Step 5 (Linearization).}
Multiply the display in Step~4 by $\sqrt{|B_k|}$, use the stochastic order 
convention $O_p(1)\cdot o_p(|B_k|^{-1/2}) = o_p(|B_k|^{-1/2})$, and apply 
Slutsky's theorem using $\tilde A_k^{-1}\to_p A^{-1}$ from Step~3:
\[
\sqrt{|B_k|}\,(\hat\theta_k-\theta_0)
=
-A^{-1}\cdot \frac{1}{\sqrt{|B_k|}}\sum_{t\in B_k}\psi\!\big(W_t;\theta_0,\eta_0\big)
\;+\; o_p(1).
\]
Uniformity over $k$ holds because $K$ is fixed and all $o_p(\cdot)$ terms above 
are uniform in $k$. \hfill $\square$
\end{proof}
\subsection*{Proof of Theorem \ref{thm:oracle}}
\begin{proof}
\textbf{Step 1 (Linearization).}
From Lemma~\ref{lem:linear}, for each fold $k$,
\[
\sqrt{|B_k|}(\hat\theta_k-\theta_0)
= - A^{-1}\cdot \frac{1}{\sqrt{|B_k|}}\sum_{t\in B_k}\psi(W_t;\theta_0,\eta_0) + o_p(1),
\]
where the $o_p(1)$ term is uniform in $k$. Equivalently,
\[
\hat\theta_k-\theta_0
= -A^{-1}\cdot \frac{1}{|B_k|}\sum_{t\in B_k}\psi(W_t;\theta_0,\eta_0) + o_p(|B_k|^{-1/2}).
\]

\textbf{Step 2 (Average over folds).}
Taking the average over $k=1,\ldots,K$,
\[
\hat\theta-\theta_0
= -A^{-1}\cdot \frac{1}{K}\sum_{k=1}^K \frac{1}{|B_k|}\sum_{t\in B_k}\psi(W_t;\theta_0,\eta_0)
+ \frac{1}{K}\sum_{k=1}^K o_p(|B_k|^{-1/2}).
\]
Since $K$ is fixed and the remainders are uniform in $k$ (Lemma~\ref{lem:linear}), 
and $|B_k| \asymp T$ for each $k$, we have
\[
\frac{1}{K}\sum_{k=1}^K o_p(|B_k|^{-1/2}) = o_p(T^{-1/2}).
\]
Define $\bar m_k := |B_k|^{-1}\sum_{t\in B_k}\psi(W_t;\theta_0,\eta_0)$. Then
\[
\hat\theta-\theta_0
= -A^{-1}\cdot \frac{1}{K}\sum_{k=1}^K \bar m_k + o_p(T^{-1/2}).
\]

\textbf{Step 3 (Block average equals full-sample average).}
Write
\[
\frac{1}{K}\sum_{k=1}^K \bar m_k
= \sum_{k=1}^K \frac{|B_k|}{T}\,\bar m_k 
+ \sum_{k=1}^K\left(\frac{1}{K}-\frac{|B_k|}{T}\right)\bar m_k
=: (I) + (II).
\]
For term $(I)$: Since the blocks partition $\{1,\ldots,T\}$ up to $O(1)$ edge effects (Assumption~\ref{ass:blocks}),
\[
(I) = \frac{1}{T}\sum_{t=1}^T\psi(W_t;\theta_0,\eta_0) + O_p(T^{-1}).
\]
For term $(II)$: By construction (Assumption~\ref{ass:blocks}), $|B_k| = \lfloor T/K\rfloor + O(1) = T/K + O(1)$, so
\[
\left|\frac{1}{K} - \frac{|B_k|}{T}\right| 
= \frac{\big||B_k| - T/K\big|}{T} 
= \frac{O(1)}{T}
= O(T^{-1}).
\]
By the invariance principle (Assumption~\ref{ass:FCLT}) applied to each block and the mean-zero property $\E[\psi(W_t;\theta_0,\eta_0)]=0$, 
\[
\|\bar m_k\| = O_p(|B_k|^{-1/2}) = O_p(T^{-1/2}).
\]
Therefore, with $K$ fixed,
\[
\|(II)\| \le \sum_{k=1}^K \left|\frac{1}{K} - \frac{|B_k|}{T}\right| \cdot \|\bar m_k\|
= K \cdot O(T^{-1}) \cdot O_p(T^{-1/2}) = O_p(T^{-3/2}) = o_p(T^{-1/2}).
\]
Combining, 
\[
\frac{1}{K}\sum_{k=1}^K \bar m_k = \frac{1}{T}\sum_{t=1}^T\psi(W_t;\theta_0,\eta_0) + o_p(T^{-1/2}).
\]

\textbf{Step 4 (Asymptotic normality).}
We have shown
\[
\sqrt{T}(\hat\theta-\theta_0)
= -A^{-1}\cdot \frac{1}{\sqrt{T}}\sum_{t=1}^T\psi(W_t;\theta_0,\eta_0) + o_p(1).
\]
By Assumption~\ref{ass:FCLT}, the functional central limit theorem at $u=1$ gives
\[
\frac{1}{\sqrt{T}}\sum_{t=1}^T\psi(W_t;\theta_0,\eta_0) \Rightarrow \mathcal{N}(0,\Sigma).
\]
Since $A$ is nonsingular, 
by Slutsky's theorem and the continuous mapping theorem,
\[
\sqrt{T}(\hat\theta-\theta_0) \Rightarrow -A^{-1}\cdot \mathcal{N}(0, \Sigma)= \mathcal{N}(0, (-A^{-1})\Sigma (-A^{-1})^\top)=\mathcal{N}(0, A^{-1}\Sigma (A^{-1})^\top).
\] \hfill $\square$
\end{proof}

\end{document}